\documentclass[%
reprint,
superscriptaddress,
amsmath,amssymb,
aps,
pra,
onecolumn,
]{revtex4-2}

\usepackage[pdftex]{graphicx,color}
\usepackage{dcolumn}
\usepackage{bm}

\usepackage{physics}
\usepackage{braket}
\usepackage{float}
\usepackage{siunitx}
\usepackage{amsmath,mathtools,amssymb,enumerate,here,amsthm}
\newtheorem{prop}{Proposition}

\begin{document}

\preprint{APS/123-QED}

\title{Standard quantum limit of finite-size optical lattice clock in estimating \\gravitational potential}

\author{Fumiya Nishimura}
 \email{fumiya.nishimura562@s.kyushu-u.ac.jp}
 \affiliation{Department of Physics, Kyushu University, 744 Motooka, Nishi-Ku, Fukuoka 819-0395, Japan}
\author{Yui Kuramochi}
 \email{kuramochi.yui@phys.kyushu-u.ac.jp}
 \affiliation{Department of Physics, Kyushu University, 744 Motooka, Nishi-Ku, Fukuoka 819-0395, Japan}
\author{Kazuhiro Yamamoto}
 \email{yamamoto@phys.kyushu-u.ac.jp}
 \affiliation{Department of Physics, Kyushu University, 744 Motooka, Nishi-Ku, Fukuoka 819-0395, Japan}
 \affiliation{Research Center for Advanced Particle Physics, KyushuUniversity, 744 Motooka, Nishi-ku, Fukuoka 819-0395, Japan}
 \affiliation{International Center for Quantum-Field Measurement Systems for Studies of the Universe and Particles (QUP), KEK, Oho 1-1, Tsukuba, Ibaraki 305-0801, Japan}
 




\begin{abstract}
  We evaluated the accuracy limit for estimating gravitational potential using optical lattice clocks by utilizing the quantum Cram\'{e}r--Rao bound.
  We then compared the results for single-layer and multilayer optical lattice clocks. 
  The results indicate that the lower bound of variance of the estimator of gravitational potential using finite-size optical lattice clocks diverges and recovers repeatedly as a function of time. 
  Namely, the accuracy of the gravitational potential estimation is not a monotonic function of time owing to the effect of gravitational dephasing in finite-size optical lattice clock. 
  Further,  this effect creates an estimation accuracy limit when attempting to avoid the divergence of the lower bound.
  When the number of layers in the optical lattice clock is sufficiently large, the limit is independent of the optical lattice clock details.
  The time required to reach this limit is calculated to be approximately 33 hours for a three-dimensional optical lattice clock consisting of one million cadmium atoms due to Earth's gravity, and approximately the same for other atoms.
\end{abstract}

\maketitle


\section{Introduction}
Recently, the accuracy of time measurements using optical lattice clocks has significantly improved.
A pertinent example is the optical lattice clock used in the experiment at the Tokyo Skytree.
This optical lattice clock is portable and accurate on the order of $10^{-18}$~\cite{Takamoto_2020}.
This is on the order of $\SI{e-2}{\meter}$ in terms of the height difference near the Earth's surface, and the clock can detect height differences of several centimeters.
This renders the use of optical lattice clocks as high-accuracy gravitational potential meters in geodetic applications feasible.
In addition to pure geodetic applications, applications in seismology and volcanology for monitoring crustal deformations have also been considered~\cite{Tanaka_2021}.
The accuracy of optical lattice clocks will be further improved in the future.
For example, optical lattice clocks accurate on the orders of $10^{-19}$ and $10^{-20}$ were demonstrated by Zheng et al.~\cite{Zheng_2022} and Bothwell et al.~\cite{Bothwell_2022}, respectively.
Pedrozo-P\~{e}nafiel et al.~\cite{Pedrozo_2020} reported that an accuracy beyond the standard quantum limit (SQL) can be achieved upon using spin-squeezed states.

By contrast, gravitational dephasing has been shown to reduce the accuracy of optical lattice clocks~\cite{Kawasaki_2022}.
In a vertically layered optical lattice clock, the gravitational potentials of the atomic clocks in each layer are different, resulting in a phase difference due to the gravitational redshift and loss of coherence in the system as a whole.
This is referred to as gravitational dephasing, as described by Kawasaki~\cite{Kawasaki_2022}.
This should be discriminated from the gravitational decoherence~\cite{Pikovski_2015} or the loss of coherence due to the gravitational interaction between atomic clocks~\cite{Esteban_2017}.
When measuring the gravitational potential using an optical lattice clock, the accuracy of the measurement is expected to deteriorate because of gravitational dephasing.

The aforedescribed prior findings indicate that the effect of gravitational dephasing may not be negligible when the accuracy of optical lattice clocks is further improved or when new application methods are devised and require further accuracy.
This suggests that evaluating the gravitational effect on the accuracy would be worthwhile.
In this study, we evaluated the quantum Fisher information of the statistical model of an optical lattice clock in a gravitational field. 
According to the quantum Cram\'{e}r--Rao bound, the inverse of the quantum Fisher information gives a lower bound of the accuracy of any unbiased estimator of the gravitational potential.
The effect of gravitational dephasing can be evaluated by specifically determining the quantum Cram\'{e}r--Rao bound for estimating the gravitational potential using optical lattice clocks.

The remainder of this paper is organized as follows.
In Sec.~\ref{sec:Hamiltonian}, we present the formulation of the Hamiltonian of atomic clocks.
In Sec.~\ref{sec:OneClock}, we detail the calculation of the quantum Fisher information for the statistical model of gravitational potential estimation using a single-layer optical lattice clock.
Sec.~\ref{sec:NClocks} describes the evaluation of the Fisher information of gravitational potential estimation using an $N_{\mathrm{layer}}$-layer optical lattice clock.
In Sec.~\ref{sec:Discussion}, we discuss the results obtained and compare them with the actual values.
Sec.~\ref{sec:Summary} summarizes this study and the conclusions drawn.
In Appendix A, the classical and quantum estimation theories are briefly reviewed.
In Appendix B, an example of a positive operator-valued measure (POVM) that achieves the equality in Eq.~\eqref{eq:CRBOneClock} is presented.

\section{Time evolution of an atomic clock under gravitational field}\label{sec:Hamiltonian}
This section details the derivation of the Hamiltonian of an atomic clock in a weak gravitational field.
Our starting point is the following Hamiltonian of a composite system (a two-level atom in this study) with an internal degree of freedom moving with nonrelativistic velocity and sufficiently small acceleration in a weak gravitational field~\cite{Cepollaro_2023}:
\begin{equation}
  H= mc^2+\frac{p^2}{2m}+mV_N(x)+E_{\mathrm{int}}\left(1+\frac{V_N(x)}{c^2}-\frac{p^2}{2m^2c^2}\right).
  \label{eq:ClockFreeHamiltonian}
\end{equation}
Here, $m$ is the rest mass of the system without considering its internal energy, $x$ is the distance between the center of mass of the system and the gravitational source, $p$ is the momentum of the center of mass of the system, $V_N(x)$ is the Newtonian potential, and $E_{\mathrm{int}}$ is the internal energy of the system. 
We quantize the internal energy $E_{\mathrm{int}}$ to the internal Hamiltonian operator $\hat{H}_{\mathrm{int}}$ and assume that $\hat{H}_{\mathrm{int}}$ can be written in the form of the following two-level system:
\begin{equation}
  \hat{H}_{\mathrm{int}}=E_0\ketbra{0}{0}+E_1\ketbra{1}{1} \quad (E_0 < E_1).
  \label{eq:ClockInternalHamiltonian} 
\end{equation}
The first three terms on the RHS in Eq.~\eqref{eq:ClockFreeHamiltonian} are the rest energy, kinetic energy of the center of mass, and Newtonian potential, respectively:
This part of the Hamiltonian describes the center-of-mass system of an atom.
The fourth term is the Hamiltonian of the atomic internal degrees of freedom.

We have the relativistic correction terms ${V_N(x)}/{c^2}-{p^2}/({2m^2c^2})$ in the fourth term. 
The first term ${V_N(x)}/{c^2}$ describes the time delay due to the gravitational redshift, whereas the second term $-{p^2}/({2m^2c^2})$ originates from the time delay due to a special relativistic effect. 
The reason for this correction is that Eq.~\eqref{eq:ClockInternalHamiltonian} is only an internal Hamiltonian as observed by an observer who is stationary with respect to the atom.
Under general circumstances, the atom is not necessarily stationary from the observer's coordinate system; therefore, the time delay due to special and general relativistic effects is not necessarily $0$.

We now set $p = 0$.
This is justified because the center of mass of an atom in the optical lattice clock is confined to an extremely narrow region by, for example, an outer potential.
Then, by using Eq.~\eqref{eq:ClockInternalHamiltonian}, we obtain
\begin{equation}
  \hat{H} = mc^2+mV_N(x)+\left(\bar{E}\hat{I}+\frac{\Delta E}{2}\hat{\sigma}^z\right)\left(1+\frac{V_N(x)}{c^2}\right),
  \label{eq:ClockFixHamiltonian}
\end{equation}
where
\begin{equation}
  \begin{aligned}
    \bar{E}&=\frac{E_0+E_1}{2}, \\
    \Delta E&=E_0-E_1 ,
  \end{aligned}
\end{equation}
\begin{equation}
  \begin{aligned}
    \hat{\sigma}^z=\ketbra{0}{0}-\ketbra{1}{1},
  \end{aligned}
\end{equation}
and $\hat{I} = \ketbra{0}{0} + \ketbra{1}{1}$ denote an identity operator.
By discarding the constant terms that do not contribute to the time evolution of the system, we obtain
\begin{equation}
  \hat{H} = \frac{\Delta E}{2}\hat{\sigma}^z\left(1+\frac{V_N(x)}{c^2}\right)= \frac{\Delta E}{2} \theta_0 \hat{\sigma}^z ,
  \label{eq:ClockFixHamiltonianNonConst}
\end{equation}
where we define
\begin{align}
  \theta_0 \equiv 1+\frac{V_N(x)}{c^2}.
  \label{eq:DefTheta}
\end{align}
It must be noted that the aim is to estimate the {\it classical} gravitational potential $V_N(x)$.

\section{Accuracy of Gravitational Potential Estimation for Single-Layer Optical Lattice Clock}\label{sec:OneClock}
In this section, we describe the evaluation of the symmetric logarithmic derivative (SLD) Fisher information of the single-layer optical lattice clock shown in Fig.~\ref{fig:Model1Setting}.
The formulation of the quantum Fisher information and quantum Cram\'{e}r--Rao bounds are detailed in Appendix~\ref{app:FormQFIandQCRB}.
\begin{figure}[t]
  \centering
  \includegraphics[width=14cm]{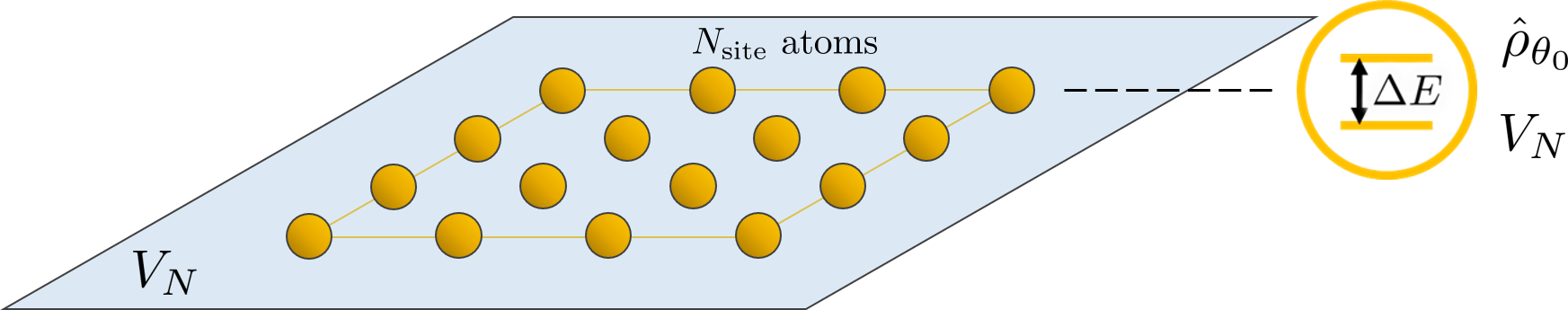}
  \vspace{1cm}
  \caption{Diagram of the system considered in Sec.~\ref{sec:OneClock}. $N_{\mathrm{site}}$ atomic clocks are captured on the plane of gravitational potential $V_N$ (blue plane).}
  \label{fig:Model1Setting}
\end{figure}
If the initial state of the atomic clock is $ \ket{\psi_0} \equiv \left(\ket{0}+\ket{1}\right)/{\sqrt{2}}$, the density operator $\hat{\rho}_{\theta_0}$ of the atomic clock after time $\tau$ is
\begin{align}
  \hat{\rho}_{\theta_0} &= e^{-i\hat{H}\tau/\hbar}\ketbra{\psi_0}{\psi_0} e^{i\hat{H}\tau/\hbar}  \notag\\
                        &=\frac{1}{2}\left\lbrace\mathbb{I} +\exp\left(-i\frac{\Delta E \theta_0}{\hbar}\tau\right)\ketbra{0}{1}+\exp\left(i\frac{\Delta E \theta_0}{\hbar}\tau\right)\ketbra{1}{0}\right\rbrace.
  \label{eq:ClockDensityOperator}
\end{align}
The SLD $\hat{L}_{\theta_0}$ is defined as follows:
\begin{equation}
  \pdv{\hat{\rho}_{\theta_0}}{\theta_0}=\frac{1}{2}\left(\hat{\rho}_{\theta_0}\hat{L}_{\theta_0}+\hat{L}_{\theta_0}\hat{\rho}_{\theta_0}\right).
   \label{eq:DefSLD}
\end{equation}
The solution of Eq.~\eqref{eq:DefSLD} is expressed as
\begin{align}
  \hat{L}_{\theta_0}=-i\frac{\Delta E}{\hbar}\tau\left\lbrace\exp\left(-i\frac{\Delta E \theta_0}{\hbar}\tau\right)\ketbra{0}{1}-\exp\left(i\frac{\Delta E \theta_0}{\hbar}\tau\right)\ketbra{1}{0}\right\rbrace.
\end{align}
Thus, the SLD Fisher information is
\begin{align}
  S =\tr[\hat{\rho}_{\theta_0} \hat{L}_{\theta_0}^2]= \left(\frac{\Delta E \tau}{\hbar}\right)^2.
\end{align}
Therefore, if the variance of the estimator $\theta^{\mathrm{est}}_0$ of $\theta_0$ is $\mathrm{Var}\left[\theta_0 ^{\mathrm{est}}\right]$ and $N_{\mathrm{site}}$ is the number of atomic clocks in the same layer, the quantum Cram\'{e}r--Rao bound gives
\begin{align}
  \mathrm{Var}\left[\theta_0 ^{\mathrm{est}}\right] &\geq \frac{1}{N_{\mathrm{site}}S} = \frac{1}{N_{\mathrm{site}}}\left(\frac{\hbar}{\Delta E \tau}\right)^2. \label{eq:CRBOneClockTheta0}
\end{align}
From the definition of $\theta_0$ in Eq.~\eqref{eq:DefTheta}, we obtain
\begin{equation}
  \frac{\mathrm{Var}\left[V_N^{\mathrm{est}}\right]}{c^4}
  \geq \frac{1}{N_{\mathrm{site}}}\left(\frac{\hbar}{\Delta E \tau}\right)^2,
  \label{eq:CRBOneClock}
\end{equation}
where $V_N^{\mathrm{est}}$ denotes the estimator of gravitational potential $V_N$.
This bound indicates that, in the absence of gravitational dephasing, the accuracy of the estimation increases in time.
The lower bound of ${\mathrm{Var}\left[V_N^{\mathrm{est}}\right] N_{\mathrm{site}}}/{c^4}$ in Eq.~\eqref{eq:CRBOneClock} is plotted as a function of ${\Delta E \tau}/{\hbar}$ in Fig.~\ref{fig:VarVModel1}.
This value has the same time dependence as the square of the standard quantum limit
\begin{align}
  \sigma(\tau,\tau_{\mathrm{avg}})=\frac{1}{\omega_0 \tau}\sqrt{\frac{T_\mathrm{C}}{\tau_{\mathrm{avg}}}}\sqrt{\frac{\xi_\mathrm{W}^2}{N_\mathrm{site}}},\quad \xi_\mathrm{W}^2=1
\end{align}
in the study by Pedrozo-P\~{e}nafiel et al.~\cite{Pedrozo_2020}; here, $\tau$ is the interrogation time, $T_\mathrm{C}$ is the clock cycle time, $\tau_{\mathrm{avg}}$ is the averaging time, $N_\mathrm{site}$ is the number of independent samples, $\omega_0=\Delta E/\hbar$ is the angular frequency of the clock transition, and $\xi_\mathrm{W}^2$ is the Wineland parameter. 
\begin{figure}[t]
  \centering
  \includegraphics[width=13cm]{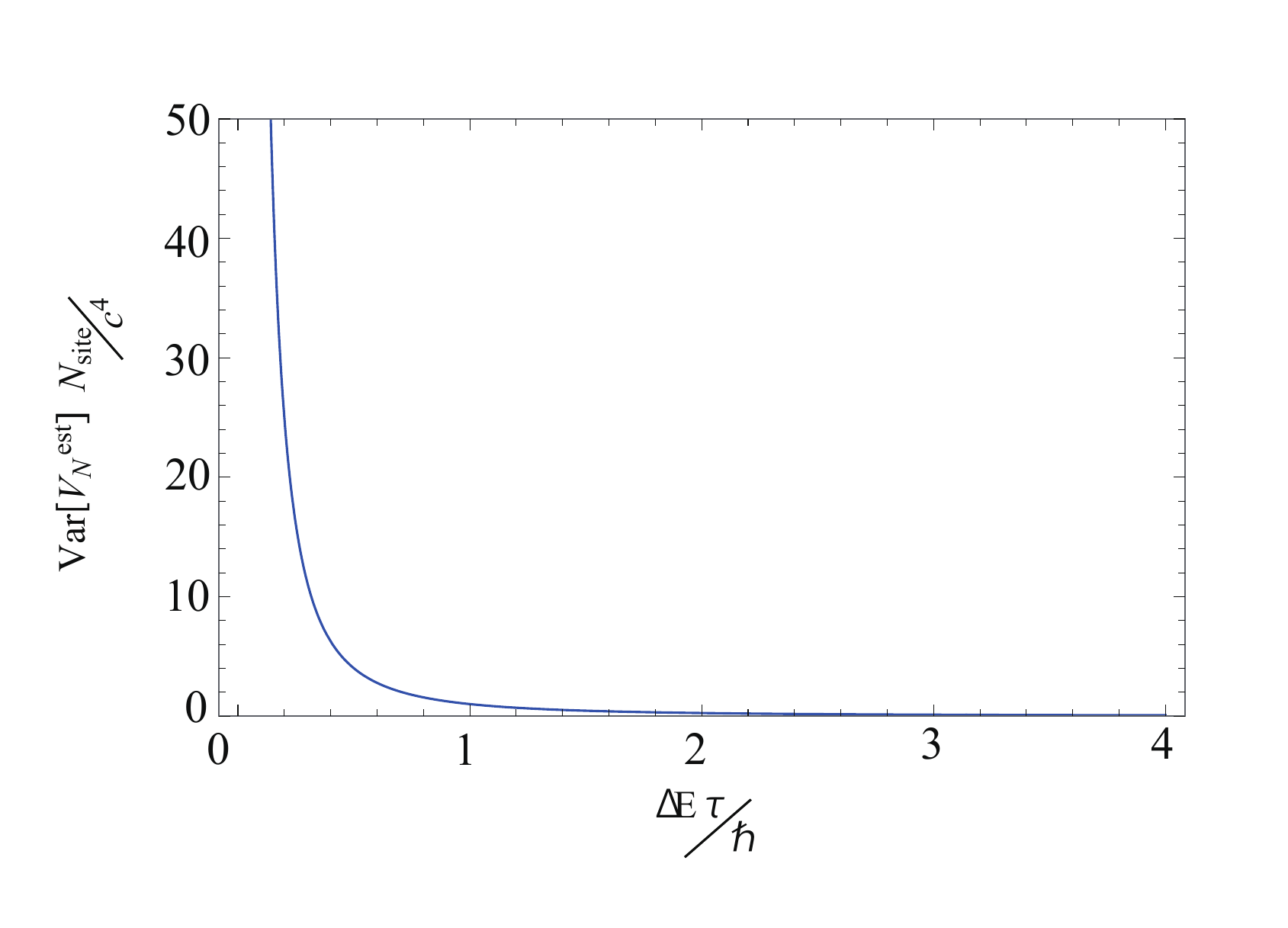}
  \vspace{-0.5cm}
  \caption{Lower bound of ${\mathrm{Var}\left[V_N^{\mathrm{est}}\right]N_{\mathrm{site}}/c^4}$ in Eq.~\eqref{eq:CRBOneClock} plotted as a function of ${\Delta E \tau}/{\hbar}$ in the single-layer scenario.}
  \label{fig:VarVModel1}
\end{figure}

\section{Accuracy of gravitational potential estimation for $N_{\mathrm{layer}}$-layered optical lattice clock}\label{sec:NClocks}
This section presents the calculation of the estimation accuracy of the gravitational potential when an $N_{\mathrm{layer}}$-layer optical lattice clock is used to estimate the gravitational potential.
Specifically, after calculating the SLD Fisher information of the atomic clock, which is indistinguishable from each of the $N_{\mathrm{layer}}$-layers vertically aligned, we apply the quantum Cram\'{e}r--Rao bound, assuming that there are $N_{\mathrm{site}}$ sets, to evaluate the variance of the estimated gravitational potential (see Fig.~\ref{fig:Model2Setting}).
This setup is based on the assumption that the observer can only perform measurements on single atoms and cannot distinguish the atoms being measured. 
The effective density operator of the entire system can then be written as the ensemble average of the density operators of the atomic clocks.
Upon further assuming that each layer has the same population of atoms, the single-atom density operator can be obtained as
\begin{equation}
  \hat{\rho}_{\theta_0}=\frac{1}{2\ell+1}\sum_{j=-\ell}^\ell \hat{\rho}_j,
  \label{eq:ClockDensityOperatorMixedBefCal}
\end{equation}
where $\rho_j$ is the density operator of an atom in the $j$th layer, $j$ is the label of a layers running from $-\ell$ to $\ell$, and $\theta_0 \equiv 1 + V_0 / c^2$.
Here, $N_{\mathrm{layer}}=2\ell+1$ denotes the total number of layers.
The density operator in Eq.~\eqref{eq:ClockDensityOperatorMixedBefCal} can be derived as follows.
Suppose that the observer performs the measurement corresponding to the single-atom POVM $(\hat{E}(x))$ but does not distinguish which atom is measured.
The probability of obtaining outcome $x$ when the state is $\hat{\rho}_j$ is
\begin{equation}
	P(x|j) = \tr[\hat{\rho}_j \hat{E}(x)].
\end{equation}
Because each $\hat{\rho}_j$ is measured with an equal probability $1/N_{\mathrm{layer}},$ the probability of obtaining $x$ is given by
\begin{equation}
	P(x) = \sum_{j=-\ell}^\ell P(x|j) \times \frac{1}{N_{\mathrm{layer}}} = \tr[\hat{\rho}_{\theta_0} \hat{E}(x)].
\end{equation}
Noting that the POVM $\hat{E}$ is arbitrary, we conclude that the system state is effectively described by $\hat{\rho}_{\theta_0}$.

We assume that the gravitational field is uniform, gravitational acceleration is $g$, distance between layers is $h$, and gravitational potential in the $0$th layer is $V_0$.
Then, on the basis of Eq.~\eqref{eq:ClockFixHamiltonianNonConst}, the gravitational potential and the Hamiltonian of the atomic clock in the $j$th layer can be written as
\begin{align}
  V_j=V_0+gjh
\end{align}
and
\begin{align}
  \hat{H}_j=\frac{\Delta E}{2}\left(1+\frac{V_0}{c^2}+\frac{gjh}{c^2}\right)\hat{\sigma}^z = \frac{\Delta E}{2}\left(\theta_0+j\alpha\right)\hat{\sigma}^z,
\end{align}
where $\alpha \equiv \frac{gh}{c^2}$.
If the initial state of each atomic clock is selected to be the superposition state $\frac{1}{\sqrt{2}}\left(\ket{0}+\ket{1}\right)$, the density operator $\hat{\rho}_j$ after time $\tau$ is given by
\begin{align}
  \hat{\rho}_j=\frac{1}{2}\left[\mathbb{I} +\exp\left(-\frac{i\Delta E \tau}{\hbar}\left(\theta_0+j\alpha\right)\right)\ketbra{0}{1}+\exp\left(\frac{i\Delta E \tau}{\hbar}\left(\theta_0+j\alpha\right)\right)\ketbra{1}{0}\right].
\end{align}
Therefore, on the basis of Eq.~\eqref{eq:ClockDensityOperatorMixedBefCal}, we have
\begin{align}
  \hat{\rho}_{\theta_0} &= \frac{1}{2}\left[\mathbb{I} +\frac{1}{2\ell+1}\sum_{j=-l}^{l}e^{-iA(\theta_0+j\alpha)}\ketbra{0}{1}+\frac{1}{2\ell+1}\sum_{j=-l}^{l}e^{iA(\theta_0+j\alpha)}\ketbra{1}{0}\right] \notag\\
                        &= \frac{1}{2}\mathbb{I} +\frac{1}{2\ell+1}\frac{e^{-iA\theta_0}}{2}\frac{\sin\left[\frac{A\alpha}{2}(2\ell+1)\right]}{\sin\left[\frac{A\alpha}{2}\right]}\ketbra{0}{1}+\frac{1}{2\ell+1}\frac{e^{iA\theta_0}}{2}\frac{\sin\left[\frac{A\alpha}{2}(2\ell+1)\right]}{\sin\left[\frac{A\alpha}{2}\right]}\ketbra{1}{0},
  \label{eq:ClockDensityOperatorMixed}
\end{align}
where $A\equiv\frac{\Delta E \tau}{\hbar}$.
Then, Eq.~\eqref{eq:ClockDensityOperatorMixed} can be diagonalized as 
\begin{align}
  \hat{\rho}_{\theta_0}= \frac{1}{2}\left[\left(1+\frac{1}{2\ell+1}\frac{\sin\left[\frac{A\alpha}{2}(2\ell+1)\right]}{\sin\left[\frac{A\alpha}{2}\right]}\right)\ketbra{\psi_+}{\psi_+}+\left(1-\frac{1}{2\ell+1}\frac{\sin\left[\frac{A\alpha}{2}(2\ell+1)\right]}{\sin\left[\frac{A\alpha}{2}\right]}\right)\ketbra{\psi_-}{\psi_-}\right], \label{eq:ClockDensityOperatorMixedDiag} 
\end{align}
\begin{align}
  \ket{\psi_\pm}   \equiv \frac{1}{\sqrt{2}}\left(\ket{0}\pm e^{iA\theta_0}\ket{1}\right). \label{eq:DefPsiPM}
\end{align}
\begin{figure}[t]
  \centering
  \includegraphics[width=15cm]{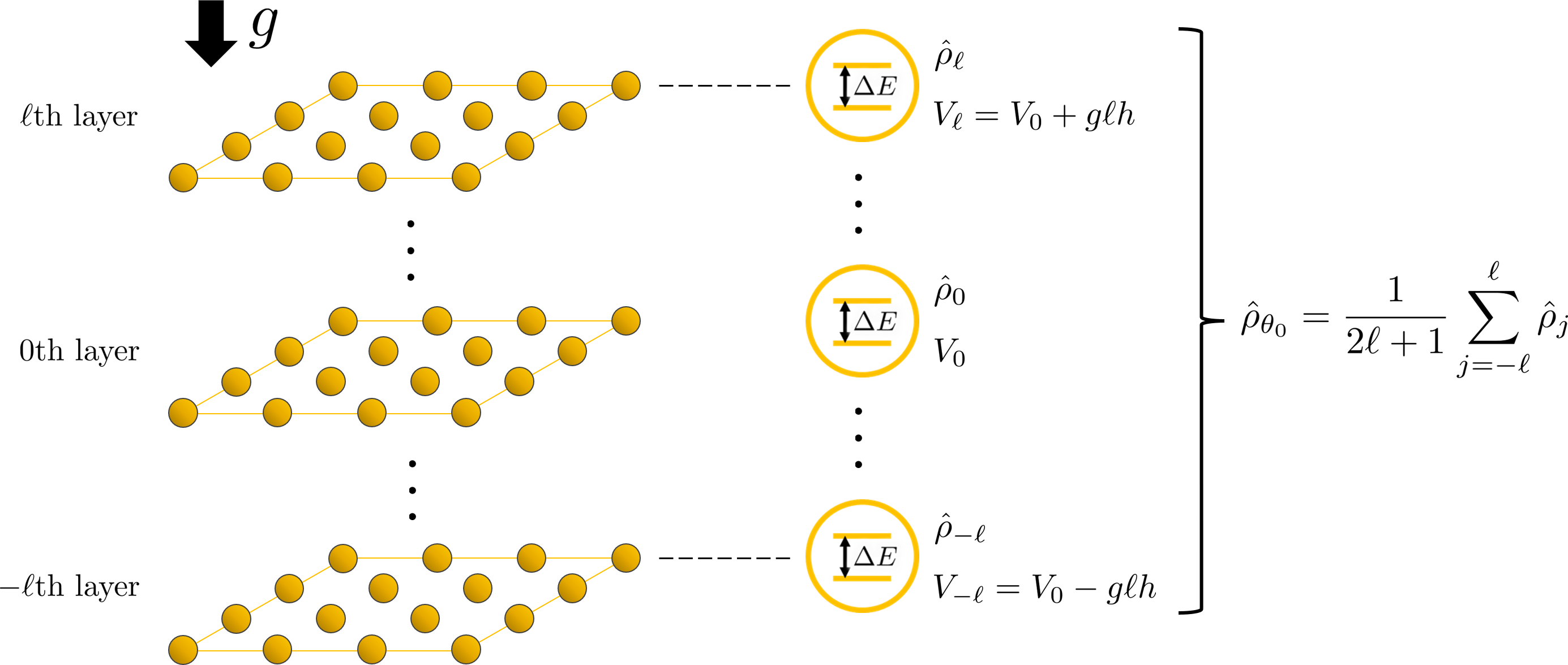}
  \vspace{0.5cm}
  \caption{Diagram of the system considered in Sec.~\ref{sec:NClocks}. The $N_{\mathrm{site}}$ atoms are captured on each layer and the layers are equally spaced in the vertical direction from the $-\ell$th to $\ell$th layer. We assume a uniform gravitational field, and the gravitational acceleration is $g$ and the layer spacing is $h$. To consider the situation wherein the layers are indistinguishable, the density operator $\hat{\rho}_{\theta_0}$ for the entire system is the ensemble average of the density operators for each atomic clock.}
  \label{fig:Model2Setting}
\end{figure}

Next, we evaluate the SLD Fisher information with respect to the parameter $\theta_0$.
The SLD in this case is given by 
\begin{equation}
  \hat{L}_{\theta_0}=iA\left(\frac{1}{2\ell+1}\frac{\sin\left[\frac{A\alpha}{2}(2\ell+1)\right]}{\sin\left[\frac{A\alpha}{2}\right]}\right)\left(\ketbra{\psi_+}{\psi_-}-\ketbra{\psi_-}{\psi_+}\right).
\end{equation}
Therefore, the SLD Fisher information is as follows:
\begin{align}
  S =\tr[\hat{\rho}_{\theta_0} \hat{L}_{\theta_0}^2]= \left(\frac{A}{2\ell+1}\frac{\sin\left[\frac{A\alpha}{2}(2\ell+1)\right]}{\sin\left[\frac{A\alpha}{2}\right]}\right)^2.
\end{align}
Thus, if the variance of an estimator $\theta^{\mathrm{est}}$ of $\theta_0$ is expressed as $\mathrm{Var}\left[\theta_0^{\mathrm{est}}\right]$ and the number of atomic clocks in the same layer is $N_{\mathrm{site}}$, the quantum Cram\'{e}r--Rao bound gives
\begin{align}
  \mathrm{Var}\left[\theta_0^{\mathrm{est}}\right] &\geq \frac{1}{N_{\mathrm{site}}S} = \frac{1}{N_{\mathrm{site}}}\left(\frac{2\ell+1}{A}\frac{\sin\left[\frac{A\alpha}{2}\right]}{\sin\left[\frac{A\alpha}{2}(2\ell+1)\right]}\right)^2.
\end{align}
By putting $A\equiv{\Delta E \tau}/{\hbar}$, $\theta_0 \equiv 1+{V_N(x)}/{c^2}$, and $\alpha \equiv {gh}/{c^2}$, we obtain 
\begin{align}
  \frac{\mathrm{Var}\left[V_0^{\mathrm{est}}\right]}{c^4}
  \geq \frac{1}{N_{\mathrm{site}}}\left(\frac{\hbar}{\Delta E \tau}\right)^2\left((2\ell+1)\frac{\sin\left[\frac{\Delta E gh }{2\hbar c^2}\tau\right]}{\sin\left[\frac{\Delta E gh }{2\hbar c^2}(2\ell+1)\tau\right]}\right)^2.
  \label{eq:CRBNClocks}
\end{align}
The lower bound of ${\mathrm{Var}\left[V_0^{\mathrm{est}}\right] N_{\mathrm{site}}}/{c^4}$ in Eq.~\eqref{eq:CRBNClocks} is plotted as a function of ${\Delta E \tau}/{\hbar}$ in Fig.~\ref{fig:VarVModel2}.
\begin{figure}[t]
  \centering
  \includegraphics[width=13cm]{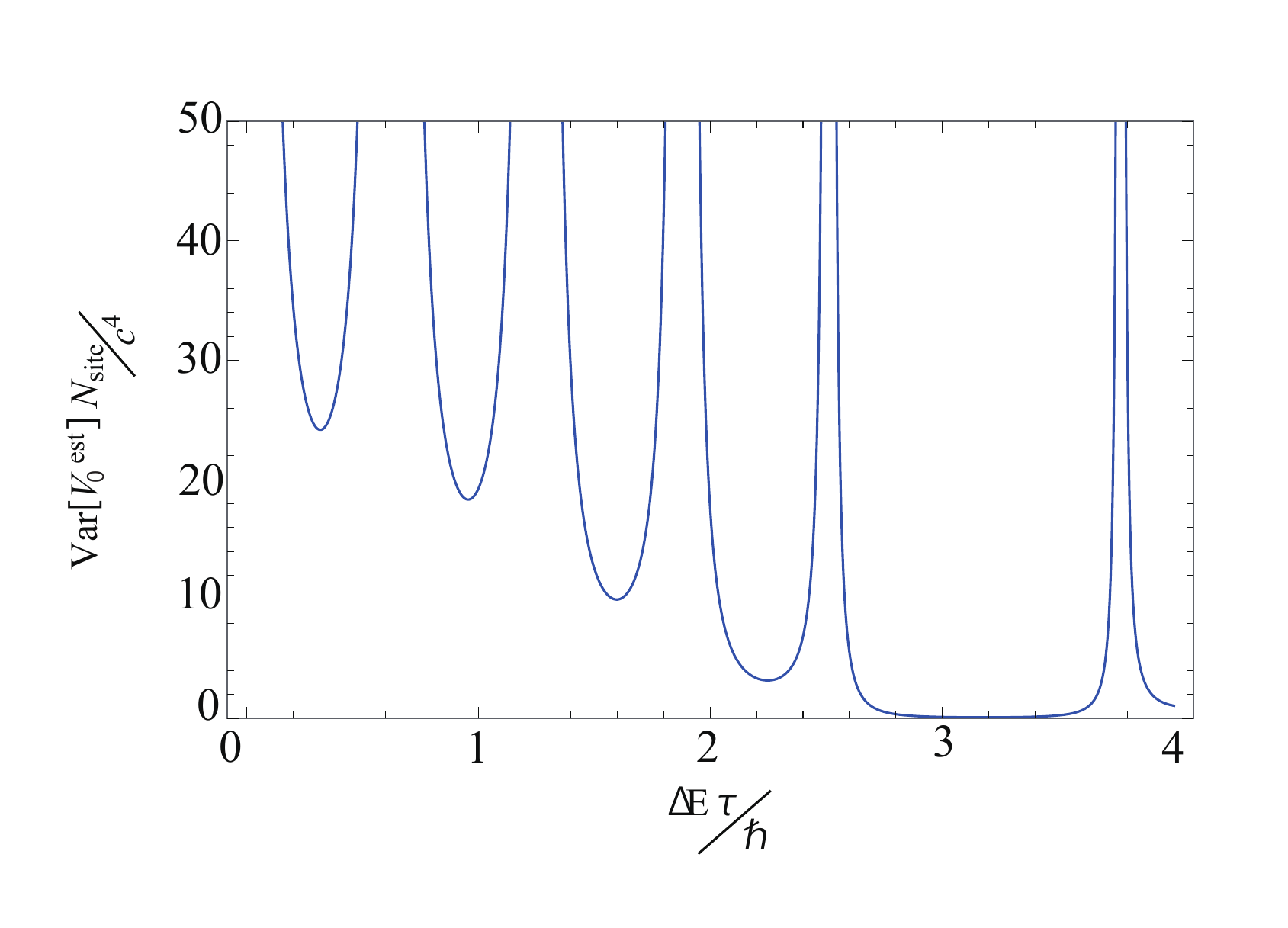}
  \vspace{-0.5cm}
  \caption{Lower bound of ${\mathrm{Var}\left[V_0^{\mathrm{est}}\right]N_{\mathrm{site}}}/{c^4}$ in Eq.~\eqref{eq:CRBNClocks} plotted as a function of ${\Delta E \tau}/{\hbar}$ in the $N_{\mathrm{layer}}$-layer scenario, where we fixed $2\ell+1=5$.}
  \label{fig:VarVModel2}
\end{figure}

\section{Discussion}\label{sec:Discussion}
The implications of the lower bounds calculated as described in Secs.~\ref{sec:OneClock} and \ref{sec:NClocks} are discussed in this section.
First, upon comparing Eqs.~\eqref{eq:CRBOneClock} and \eqref{eq:CRBNClocks}, the effect of gravitational dephasing is indicated by the factor $\left((2\ell+1){\sin\left[\frac{\Delta E gh }{2\hbar c^2}\tau\right]}/{\sin\left[\frac{\Delta E gh }{2\hbar c^2}(2\ell+1)\tau\right]}\right)^2$.
Because of this factor, the lower bound of $\mathrm{Var}\left[V_0^{\mathrm{est}}\right]$, and hence $\mathrm{Var}\left[V_0^{\mathrm{est}}\right]$ itself, diverges at time 
\begin{align}
  \tau_{\mathrm{div}}=\frac{k}{2\ell+1}\frac{2\hbar c^2}{\Delta E gh}\pi
  \label{eq:TimeDiv}
\end{align}
for integer $k$ unless ${k}/({2\ell+1})$ is an integer.
This divergence is due to the progression of dephasing in the entire system, which is caused by an increase in the phase difference between the atomic clocks in each layer.
The density operator of the system at time $\tau_{\mathrm{div}}$ is
\begin{align}
 \hat{\rho}_\theta (\tau_{\mathrm{div}})&=\frac{1}{2}\left(\ketbra{\psi_+}{\psi_+}+\ketbra{\psi_-}{\psi_-}\right) \notag\\
                                      &=\frac{1}{2}\left(\ketbra{0}{0}+\ketbra{1}{1}\right)
\end{align}
and is in a completely mixed state.
The recovery is due to further time evolution after divergence, which produces a layer wherein the phase difference is one lap behind. This resulted in a timing in which the phase coincides with the other layers.

As indicated by Fig.~\ref{fig:VarVModel2}, the lower bound of $\mathrm{Var}\left[V_0^{\mathrm{est}}\right]$ has local minimum points owing to the balance between gravitational dephasing and accuracy improvement over time; the latter improvement is shown in the single-layer scenario in Sec.~\ref{sec:OneClock}.
The point at which $\tau$ is the smallest among these local minimum points provides the principal limit of estimation accuracy when attempting to avoid divergence within a short period.
To find the local minimum in an analytic manner, the following approximation is applied:
When $2\ell+1\gg1$ and $\tau$ are within the range $0 < \tau < \frac{1}{2\ell+1}\frac{2\hbar c^2}{\Delta E gh}\pi$, we may approximate as $\sin\left[{\Delta E gh \tau}/({2\hbar c^2})\right]\backsimeq {\Delta E gh \tau}/({2\hbar c^2})$; hence,
\begin{align}
  \frac{\mathrm{Var} \left[V_0^{\mathrm{est}}\right]}{c^4}
  &\geq  \frac{1}{N_{\mathrm{site}}}\left(\frac{\hbar}{\Delta E \tau}\right)^2\left((2\ell+1)\frac{\sin\left[\frac{\Delta E gh }{2\hbar c^2}\tau\right]}{\sin\left[\frac{\Delta E gh }{2\hbar c^2 }(2\ell+1)\tau\right]}\right)^2 \notag \\
                                              &\simeq \frac{1}{N_{\mathrm{site}}}\left(\frac{gh}{2c^2}\right)^2\left(\frac{2\ell+1}{\sin\left[\frac{\Delta E gh }{2\hbar c^2}(2\ell+1)\tau\right]}\right)^2.
\label{eq:CRBNClocksApp}
\end{align}
Eq.~\eqref{eq:CRBNClocksApp} yields a local minimum value in the range $\tau >0$ at
\begin{align}
  \tau_{\mathrm{min}}=\frac{1}{2\ell+1}\frac{\hbar c^2}{\Delta E gh}\pi,
  \label{eq:TimeMin}
\end{align}
and the lower bound of $ \mathrm{Var}\left[V_0^{\mathrm{est}}\right]/c^4$ at $\tau_{\mathrm{min}}$ is 
\begin{align}
  \frac{\mathrm{Var}\left[V_0^{\mathrm{est}}\right]}{c^4}
  \geq \frac{1}{N_{\mathrm{site}}}\left[\frac{g(2\ell+1)h}{2c^2}\right]^2.
  \label{eq:CRBNClocksAppMin}
\end{align}
Therefore, when the number of layers in the optical lattice clock is sufficiently large, the lower bound of standard deviation in principle, while seeking to avoid divergence, is $\frac{1}{\sqrt{N_{\mathrm{site}}}}\frac{g(2\ell+1)h}{2c^2}$.
This can be interpreted from two perspectives.
First, it must be noted that $(2\ell+1)h$ is the height of the optical lattice clock.
In this case, $g(2\ell+1)h$ is the gravitational potential difference $\Delta V$ between the top and bottom of the optical lattice clock. Eq.~\eqref{eq:CRBNClocksAppMin} implies that
\begin{align}
  \frac{\sigma\left(V_0^{\mathrm{est}}\right)}{c^2}
  \geq \frac{1}{\sqrt{N_{\mathrm{site}}}}\frac{\Delta V}{2c^2},
  \label{eq:CRBNClocksStaDev}
\end{align}
where $\sigma(\cdot)$ denotes standard deviation.
The RHS of Eq.~\eqref{eq:CRBNClocksStaDev} does not depend on the optical lattice clock details.
Second, it must be noted that $2\ell+1$ is the number of layers of the optical lattice clock, $N_{\mathrm{layer}}$.
In this case, the standard deviation is
\begin{align}
  \frac{\sigma\left(V_0^{\mathrm{est}}\right)}{c^2}
  \geq \frac{N_{\mathrm{layer}}}{\sqrt{N_{\mathrm{site}}}}\frac{gh}{2c^2}.
\end{align}
If the atoms are equally spaced horizontally and vertically, ${N_{\mathrm{layer}}}/{\sqrt{N_{\mathrm{site}}}}$ can be regarded as the aspect ratio of the optical lattice clock.
In particular, if the optical lattice clock has a shape with vertical and horizontal symmetry, such as a cube, we have ${N_{\mathrm{layer}}}/{\sqrt{N_{\mathrm{site}}}}=1$; therefore, the standard deviation depends only on the interlayer distance $h$.

The estimation accuracy begins to deteriorate when the interrogation time exceeds the time $\tau_{\mathrm{min}}=\frac{1}{2\ell+1}\frac{\hbar c^2}{\Delta E gh}\pi$.
We now evaluate the specific value of $\tau_{\mathrm{min}}$ assuming a three-dimensional optical lattice clock.
First, we calculate for Cd atoms.
Because the wavelength is used as a clock for the Cd atoms, $\lambda^{\mathrm{Cd}}_{\mathrm{clock}}=\SI{332}{\nano\meter}$~\cite{Yamaguchi_2019}, $\Delta E=\SI{6.0e-19}{\joule}$.
In the optical lattice clock, atoms are captured in the antinode of the standing wave created by the electromagnetic wave of the corresponding magic wavelength; therefore, the spacing between each layer is approximately the magic wavelength of cadmium $\lambda^{\mathrm{Cd}}_{\mathrm{magic}}=\SI{420}{\nano\meter}$~\cite{Yamaguchi_2019} and $h=\SI{4.2e-7}{\meter}$.
Assuming that the number of atoms that can be captured in a cubic optical lattice is 1 million, the number of atoms per side is 100, and $2\ell+1=100$.
In this case, $\tau_{\mathrm{min}}$ is calculated as 
\begin{align}
  \tau_{\mathrm{min}} = \SI{1.2e5}{\second}\left(\frac{1.0\times10^2}{2\ell+1}\right)\left(\frac{\SI{6.0e-19}{\joule}}{\Delta E}\right)\left(\frac{\SI{4.2e-7}{\meter}}{h}\right).
\end{align}
This duration is approximately 33 hours.
Similarly, we calculated $\tau_{\mathrm{min}}$ for other atoms (see TABLE~\ref{tab:TauMin}).
The values are roughly the same for all atoms, including Cd.
\begin{table}[t]
  \centering
  \caption{\label{tab:TauMin}Clock wavelength and magic wavelength of Sr~\cite{Ludlow_2008}, Yb~\cite{Barber_2008}, Cd~\cite{Yamaguchi_2019}, Hg~\cite{McFerran_2012}, and Mg~\cite{Kulosa_2015} atoms. We calculated $\tau_{\mathrm{mim}}$ for each atom on the basis of these values.}
    \begin{ruledtabular}
      \begin{tabular}{lccccc}
        Atom                     &Sr               &Yb               &Cd               &Hg              &Mg\\ 
        \hline
        Clock wavelength(nm)     &698              &578              &332              &266             &458\\
        Magic wavelength(nm)     &813              &759              &420              &363             &468\\
        $\tau_{\mathrm{min}}$(s) &$1.3\times10^{5}$&$1.2\times10^{5}$&$1.2\times10^{5}$&$1.1\times10^{5}$&$1.5\times10^{5}$\\
      \end{tabular}
    \end{ruledtabular}
\end{table}

\section{Summary and Conclusion}\label{sec:Summary}
We evaluated the lower bounds of the variance of the estimators of the gravitational potential using optical lattice clocks based on the quantum Cram\'{e}r--Rao bound.
We then compared the results for the single-layer and multilayer optical lattice clocks. 
The results indicate that the lower bound of variance of the estimator using finite-size optical lattice clocks approaches zero with time in the single-layer scenario, whereas it diverges and recovers repeatedly owing to the effect of gravitational dephasing in the multilayer scenario.
This effect also produces a local minimum point, indicating that there is a limit to the estimation accuracy because of the effect of gravitational dephasing when attempting to avoid the divergence of the lower bound.
When the number of layers in the optical lattice clock is sufficiently large, the standard deviation of the estimate cannot be less than $\frac{N_{\mathrm{layer}}}{\sqrt{N_{\mathrm{site}}}}\frac{gh}{2}$. 
The interrogation time $\tau_{\mathrm{min}}$ required to reach this limit is determined to be $\frac{1}{2\ell+1}\frac{\hbar c^2}{\Delta E gh}\pi$. 
This standard deviation is independent of the optical lattice clock details.
When $N_{\mathrm{layer}}$ and $\sqrt{N_{\mathrm{site}}}$ are comparable, the standard deviation is approximately equal to the distance between the layers of the optical lattice clock.
The time $\tau_{\mathrm{min}}$ is calculated to be approximately 33 hours for a three-dimensional optical lattice clock consisting of one million Cd atoms on the Earth; the values are similar for other atoms.

In conclusion, there is certainly a limit to the lower bound of the variance of the estimators of the gravitational potential (i.e., the accuracy of gravitational potential estimation) using finite-size optical lattice clocks owing to gravitational dephasing.
The lower bound is $\frac{N_{\mathrm{layer}}}{\sqrt{N_{\mathrm{site}}}}\frac{gh}{2}$ in standard deviation.
Although this is applicable to general optical lattice clocks, it is extremely challenging to observe this limit, at least with current technology. 
This is because $\tau_{\mathrm{min}}$ is much larger than the current coherence time ($\SI{26}{\second}$) as previously reported~\cite{Zheng_2022}.

\begin{acknowledgments}
We thank Yasusada Nambu and Akio Kawasaki for crucial discussions and comments on the content of this paper. 
K.Y. was supported by JSPS KAKENHI (Grant No.~JP22H05263, No.~JP23H01175).  
Y.K. was supported by JSPS KAKENHI (Grant No.~JP22K13977).
\end{acknowledgments}

\appendix
\section{Classical and quantum estimation theory}\label{app:FormQFIandQCRB}
In this Appendix, we briefly review the classical and quantum estimation theories used in this study.
For details, see ~\cite{shao2003,PETZ_2011,helstrom1976,holevo2011,Watanabe_2014}.

A (classical) statistical model herein refers to a parameterized family $(P_\theta)_{\theta \in \Theta}$ of probability distributions (probability measures) on a fixed sample space $\Omega $ (equipped with a $\sigma$- algebra $\Sigma$).
For simplicity, we assume that the parameter set $\Theta$ is a 1-dimensional open or closed interval.
We also assume that $P_\theta$ can be written as $dP_\theta  = f_\theta d\mu$ for some $\sigma$-finite measure $\mu$ and non-negative density function $f_\theta $.
Our objective is to estimate the unknown parameter $\theta \in \Theta$ from the sample data $\omega \in \Omega $ generated according to the unknown probability distribution $P_\theta .$
The estimation is described by a (measurable) map $\hat{\theta} \colon \omega \mapsto \hat{\theta} (\omega) \in \Theta$ called an \textit{estimator}.
Occasionally, an estimator $\hat{\theta}$ is required to be \textit{unbiased}, which implies that the expected value $\mathbb{E}_\theta [\hat{\theta}] := \int_{\Omega} \hat{\theta} dP_\theta$ of the estimator coincides with the true value $\theta$ for all $\theta \in \Theta .$
A common quantitative measure of the goodness of an estimator $\hat{\theta}$ is the mean square error (MSE) $\mathbb{E}_\theta [(\hat{\theta} - \theta)^2]$, which coincides with the variance $\mathrm{Var}[\hat{\theta}] = \mathbb{E}_\theta [(\hat{\theta} - \mathbb{E}_\theta[\hat{\theta}])^2]$ when $\hat{\theta}$ is unbiased.

Under some smoothness condition on the density function $f_\theta$, we have the following (classical) Cram\'{e}r--Rao bound (e.g.,~\cite{shao2003}, Theorem~3.3):
\begin{equation}
	\mathrm{Var}[\hat{\theta}] \geq \frac{1}{I(\theta)} ,
	\label{eq:classicalCR}
\end{equation}
where 
\begin{equation}
	I(\theta) := \mathbb{E}_\theta \left[ \left( \partial_\theta \ln f_\theta  \right)^2 \right]
	\label{eq:classicalFI}
\end{equation}
is called the \textit{(classical) Fisher information}.
The classical Cram\'{e}r--Rao bound in Eq.~\eqref{eq:classicalCR} indicates that the inverse of the Fisher information provides a fundamental lower bound for the mean square error irrespective of the choice of the unbiased estimator.
An estimator $\hat{\theta}$ that attains the equality of the Cram\'{e}r--Rao bound is said to be \textit{efficient}.
If the statistical model is the product form $(P_\theta^{\times N})_{\theta \in \Theta}$, which corresponds to the independently and identically distributed (i.i.d.) $N$ samples, the Fisher information becomes $I_N(\theta) = N I(\theta)$ and the Cram\'{e}r--Rao bound gives
\begin{equation}
	\mathrm{Var}[\hat{\theta}] \geq \frac{1}{NI(\theta)} .
\end{equation}
Here, $P_\theta^{\times N}$ denotes the product measure defined on the product $\sigma$-algebra of the Cartesian product $\Omega^N$.

The assumption of unbiasedness or efficiency of the estimator is sometimes too stringent.
Occasionally, a statistical model has no unbiased/efficient estimator.
We will see in Appendix~\ref{app:ExPOVM} that the statistical model presented in Sec.~\ref{sec:OneClock} contains no unbiased estimators.
Nevertheless, we obtain a general result (\cite{shao2003}, Sec.~4.5.2) that an asymptotically efficient estimator exists under certain regularity conditions.

The classical statistical model is generalized to a quantum statistical model~\cite{helstrom1976,holevo2011}, which is a parameterized family $(\rho_\theta)_{\theta \in \Theta}$ of quantum states (density operators) in a fixed quantum system descrived by a Hilbert space $\mathcal{H}$.
For simplicity, we consider the case in which $\Theta$ is an interval and assume the finite dimensionality of $\mathcal{H}$ and smoothness of $\theta \mapsto \rho_\theta$.
As in the classical case, our objective in quantum estimation theory is to infer the parameter $\theta \in \Theta$ from measurement data.
Here appears a new problem that does not exist in the classical setting: because quantum theory prohibits us from directly perceiving unknown quantum states, we must choose an appropriate measurement to infer the parameter, which is formally described by a positive operator-valued measure (POVM).

By a POVM, we refer to a map $\mathsf{M}\colon \Sigma \to \mathcal{L}(\mathcal{H})$ defined on a $\sigma$-algebra $\Sigma$ on a set (sample space) $\Omega$ and taking values in the set $\mathcal{L}(\mathcal{H})$ of bounded operators on $\mathcal{H}$ satisfying (i) $\mathsf{M} (A) \geq 0$ $(\forall A \in \Sigma)$, (ii) $\mathsf{M} (\emptyset)=0$, $\mathsf{M}(\Omega) = \mathbb{I}$ (the identity operator on $\mathcal{H}$), and (iii) $\mathsf{M} (\bigcup_n A_n) = \sum_n \mathsf{M}(A_n)$ for any disjoint sequence $(A_n)$ in $\Sigma$.
For a POVM $\mathsf{M}$ and quantum state $\rho$, the outcome probability measure $P^\mathsf{M}_\rho \colon \Sigma \to \mathbb{R}$ is defined as
\begin{equation}
	P^\mathsf{M}_\rho (A) := \mathrm{tr}[\rho \mathsf{M}(A)] \quad (A \in \Sigma).
\end{equation}
The probability measure $P^\mathsf{M}_\rho$ describes the outcome distribution of the measurement corresponding to $\mathsf{M}$ when the state is prepared in $\rho$.

For a quantum statistical model $(\rho_\theta)_{\theta \in \Theta}$, the POVM $\mathsf{M}\colon \Sigma \to \mathcal{L}(\mathcal{H})$ induces a classical statistical model $(P^\mathsf{M}_{\rho_\theta})_{\theta \in \Theta}$ and the classical Cram\'{e}r--Rao bound in this statistical model gives
\begin{equation}
	\mathrm{Var}[\hat{\theta}] \geq \frac{1}{I_\mathsf{M}(\theta)},
	\label{eq:IM} 
\end{equation}
where $\hat{\theta}$ is an unbiased estimator of $\theta$ and $I_\mathsf{M}(\theta)$ denotes the classical Fisher information of $(P^\mathsf{M}_{\rho_\theta})_{\theta \in \Theta}$ which explicitly depends on the measurement $\mathsf{M}$.

The quantum Cram\'{e}r--Rao bound~\cite{helstrom1976,holevo2011} states that an $\mathsf{M}$-independent quantum version of the Fisher information $I_Q$, or a \textit{quantum Fisher information}, gives an upper bound of the $\mathsf{M}$-dependent classical Fisher information:
\begin{equation}
	I_\mathsf{M}(\theta) \leq I_Q(\theta) .
	\label{eq:QCR1}
\end{equation}
From Eqs.~\eqref{eq:IM} and \eqref{eq:QCR1}, we obtain an $\mathsf{M}$-independent bound of the MSE of the estimator. 
\begin{equation}
	\mathrm{Var}[\hat{\theta}] \geq \frac{1}{I_Q (\theta)}.
\end{equation}
There is an infinite amount of quantum Fisher information that generalizes the classical Fisher information~\cite{petz1996}, which originates from the noncommutativity of quantum theory.
One common choice is the symmetric logarithmic derivative (SLD) Fisher information defined as
\begin{align}
  S=\tr\left[\rho_\theta L_\theta^2 \right],
\end{align}
where $L_\theta$ is an Hermitian operator called the SLD and is defined as the solution of the following equation:
\begin{align}
 \partial_\theta \rho_\theta=\frac{1}{2} \left( \rho_{\theta} L_{\theta}+L_{\theta} \rho_{\theta} \right). 
  \label{eq:SLDdef}
\end{align}
If $\rho_\theta$ has the spectral decomposition
\begin{equation}
	\rho_\theta = \sum_j p_j \ket{j}\bra{j},
\end{equation}
then SLD $L_\theta$ has an explicit expression
\begin{equation}
	L_\theta =\sum_{j,k: p_j + p_k >0} \frac{2 \bra{j} \partial_\theta \rho_\theta \ket{k}}{p_j + p_k} \ket{j}\bra{k}.
\end{equation}
(The matrix element $\bra{j} L_\theta \ket{k}$ with $p_j = p_k = 0$ cannot be determined uniquely from Eq.~\eqref{eq:SLDdef}; however, it can be shown that such arbitrariness does not affect the value of the SLD Fisher information.) 
It is known that the SLD Fisher information is the minimal quantum Fisher information and hence gives the tightest upper bound in the quantum Cram\'{e}r--Rao bound in Eq.~\eqref{eq:QCR1}; therefore, we consider this quantity in the main section.

\section{An Example of POVM that enables efficient estimation of gravitational potential in Sec.~\ref{sec:OneClock}}\label{app:ExPOVM}
In this appendix, we provide an explicit example of a POVM that achieves the equality in Eq.~\eqref{eq:CRBOneClock}: 
We define a POVM $\mathsf{M}$ on the Borel $\sigma$-algebra $\mathcal{B}([0,2\pi))$ of the interval $[0,2\pi )$ as
\begin{align}
	d \mathsf{M}(\phi) :=\ketbra{\psi_\phi}{\psi_\phi}\frac{d\phi}{\pi} \label{eq:EffPOVM} ,\\
	\ket{\psi_\phi}=\frac{1}{\sqrt{2}}(\ket{0}+e^{i\phi}\ket{1}), \label{eq:psith}
\end{align}
where $d\phi$ denotes the Lebesgue measure.
The measurement corresponding to $\mathsf{M}$ is realized as follows. First, we randomly generate $\phi^\prime \in [0,\pi)$ according to the uniform distribution $\frac{1}{\pi} d\phi^\prime$. We then perform the projective measurement $(\ketbra{\psi_{\phi^\prime}}{\psi_{\phi^\prime}}, \ketbra{\psi_{\phi^\prime + \pi}}{\psi_{\phi^\prime+\pi}})$, and record $\phi = \phi^\prime$ (respectively, $\phi = \phi^\prime +\pi$) when the outcome of the projective measurement is $\ket{\psi_{\phi^\prime}}$ (respectively, $\ket{\psi_{\phi^\prime + \pi}}$).

The probability density function $f_{\theta_0}(\phi)$ when the state is prepared in $\hat{\rho}_{\theta_0}$, described in Sec.~\ref{sec:OneClock}, is expressed as
\begin{align}
  f_{\theta_0} (\phi)d\phi&=\tr[\hat{\rho}_{\theta_0}d\mathsf{M}(\phi)] \notag\\
                       &=\frac{d\phi}{2\pi}\left[1+\cos\left(\frac{\Delta E \tau}{\hbar}\theta_0-\phi\right)\right],
\end{align}
or
\begin{align}
  f_{\theta_0}(\phi) = \frac{1}{2\pi}\left[1+\cos\left(\frac{\Delta E \tau}{\hbar}\theta_0-\phi\right)\right].
\end{align}
Thus, the classical Fisher information is evaluated as
\begin{align}
  I(\theta_0)&= \int_0^{2\pi} (\partial_{\theta_0} \ln f_{\theta_0}(\phi))^2 f_{\theta_0}(\phi)d\phi \notag\\
             &=\left(\frac{\Delta E \tau}{\hbar}\right)^2.
\end{align}
This is equal to the lower bounds in Eq.~\eqref{eq:CRBOneClockTheta0}; therefore, there are some efficient estimators when the POVM is used in Eq.~\eqref{eq:EffPOVM}.

The classical statistical model associated with the POVM $\mathsf{M}$ has no unbiased estimator of $\theta_0$ and, hence, of the potential $V_N(x)$.
More generally, we can prove that an unbiased estimator does not exist for \textit{any} choice of measurement POVM as in the following proposition:
\begin{prop} \label{prop:ne}
Let $(\sigma_\theta)_{\theta \in \Theta}$ be a quantum statistical model defined by
\begin{gather}
	\Theta = (\alpha, \beta) \quad (-\infty<\alpha<\beta<\infty), \\
	\sigma_\theta := \ketbra{\psi_\theta}{\psi_\theta},
\end{gather}
where $\ket{\psi_\theta}$ is given by Eq.~\eqref{eq:psith}.
Then, for any POVM $\mathsf{N}$, the associated classical statistical model $(P^\mathsf{N}_{\sigma_\theta})_{\theta \in \Theta}$ has no unbiased estimator for $\theta \in \Theta$.
\end{prop}
The proof is provided in the final part of the Appendix.

As  we may write $\hat{\rho}_{\theta_0} = \sigma_{\frac{\Delta E \tau \theta_0}{\hbar}}$, Proposition~\ref{prop:ne} implies that there is no unbiased estimator for $\theta_0 .$

The classical statistical model $(P^\mathsf{M}_{\hat{\rho}_{\theta_0}})$ contains an unbiased estimator $\hat{T} \colon \phi \mapsto 2 e^{i\phi}$ of $e^{i \frac{\Delta E\tau}{\hbar}\theta_0}$.
In fact, the expectation value of $\hat{T}$ is 
\begin{align}
	\mathbb{E}_{\theta_0}[\hat{T}] &= \int_0^{2\pi} 2 e^{i\phi}f_{\theta_0} (\phi) d\phi \notag\\
	&= \int_0^{2\pi} 2e^{i\phi} \left[1+\cos\left(\frac{\Delta E \tau} {\hbar}\theta_0-\phi\right) \right] \frac{d\phi}{2\pi} \notag\\
	&= e^{i \frac{\Delta E\tau}{\hbar}\theta_0},
\end{align}
which shows the unbiasedness of $\hat{T}$.

\begin{proof}[Proof of Proposition~\ref{prop:ne}]
Let $(\Omega, \Sigma)$ be the outcome sample space (measurable space) of the POVM $\mathsf{N}$ and assume the existence of an unbiased estimator $\hat{\theta} \colon \Omega \to \mathbb{R}$ for the classical statistical model $(P^\mathsf{N}_{\sigma_\theta})_{\theta \in \Theta}$.
Then, the unbiasedness implies that
\begin{equation}
	\int_{\Omega} \hat{\theta} dP^\mathsf{N}_{\sigma_\theta}  = \theta \quad (\forall \theta \in(\alpha, \beta)).
	\label{eq:Nunb}
\end{equation}
We consider three fixed elements $\theta_1, \theta_2,\theta_3\in (\alpha, \beta)$ with $e^{i\theta_1} \ne e^{i\theta_2} \ne e^{i\theta_3} \ne e^{i\theta_1}$.
As $\sigma_\theta$ is written as
\begin{equation}
	\sigma_\theta =\frac{1}{2} [\mathbb{I} + e^{i\theta}\ketbra{1}{0}+ e^{-i\theta}\ketbra{0}{1}], 
	\label{eq:sigmath}
\end{equation}
we obtain the following linear equation:
\begin{equation}
	\begin{pmatrix}
		e^{i\theta_1}\sigma_{\theta_1} \\ e^{i\theta_2}\sigma_{\theta_2} \\ e^{i\theta_3}\sigma_{\theta_3}
	\end{pmatrix}
	=
	\begin{pmatrix}
		1 & e^{i\theta_1} & e^{2i\theta_1} \\
		1 & e^{i\theta_2} & e^{2i\theta_2} \\
		1 & e^{i\theta_3} & e^{2i\theta_3} 
	\end{pmatrix}
	\begin{pmatrix}
		\frac{1}{2} \ketbra{0}{1} \\ \frac{1}{2}\mathbb{I} \\ \frac{1}{2} \ketbra{1}{0}
	\end{pmatrix} .
	\label{eq:lineq}
\end{equation}
As the $3\times 3$ matrix on the RHS of Eq.~\eqref{eq:lineq} is invertible (Vandermonde matrix), we may write 
\begin{align}
	\frac{1}{2}\mathbb{I} &= a_1 \sigma_{\theta_1} + a_2 \sigma_{\theta_2} + a_3 \sigma_{\theta_3} \label{eq:g1} ,\\
	\frac{1}{2} \ketbra{1}{0} &= b_1 \sigma_{\theta_1} + b_2 \sigma_{\theta_2} +b_3 \sigma_{\theta_3}, \label{eq:g2} \\
	\frac{1}{2} \ketbra{0}{1} &= c_1 \sigma_{\theta_1} + c_2 \sigma_{\theta_2} +c_3 \sigma_{\theta_3} \label{eq:g3}
\end{align}
for some scalars $a_1,a_2,a_3,b_1,b_2,b_3,c_1,c_2,$ and $c_3$ that depend only on $\theta_1, \theta_2,$ and $\theta_3$.
Then, using the linearity of $ \mathcal{L}(\mathcal{H}) \ni \rho \mapsto P^\mathsf{N}_\rho $, where $P^\mathsf{N}_\rho$ for a general operator $\rho$ is a complex measure, Eqs.~\eqref{eq:Nunb}, \eqref{eq:g1}, \eqref{eq:g2}, and \eqref{eq:g3} imply that complex measures $P^\mathsf{N}_{\frac{1}{2}\mathbb{I}}$, $ P^\mathsf{N}_{\frac{1}{2}\ketbra{1}{0}}$, $P^\mathsf{N}_{\frac{1}{2}\ketbra{0}{1}}$ are written as linear combinations of the measures $P^{\mathsf{N}}_{\sigma_{\theta_1}}$, $P^{\mathsf{N}}_{\sigma_{\theta_2}}$, and $P^{\mathsf{N}}_{\sigma_{\theta_3}}$.
As the estimator $\hat{\theta}$ is integrable with respect to $P^\mathsf{N}_{\sigma_{\theta_1}}$, $P^\mathsf{N}_{\sigma_{\theta_2}}$, and $P^\mathsf{N}_{\sigma_{\theta_3}}$,
so it is with respect to the complex measures $P^\mathsf{N}_{\frac{1}{2}\mathbb{I}}$, $ P^\mathsf{N}_{\frac{1}{2}\ketbra{1}{0}}$, and $P^\mathsf{N}_{\frac{1}{2}\ketbra{0}{1}}$.
Thus, 
\begin{gather}
	\int_\Omega \hat{\theta} dP^\mathsf{N}_{\frac{1}{2}\mathbb{I}}
	= a_1\theta_1 + a_2\theta_2 + a_3 \theta_3 =: A,\\
	\int_\Omega \hat{\theta} dP^\mathsf{N}_{\frac{1}{2}\ketbra{1}{0}}
	= b_1\theta_1 + b_2\theta_2 + b_3 \theta_3 =: B,  \\
	\int_\Omega \hat{\theta} dP^\mathsf{N}_{\frac{1}{2}\ketbra{0}{1}}
	= c_1\theta_1 + c_2\theta_2 + c_3 \theta_3 =: C
\end{gather}
are well-defined.
Therefore, based on Eqs.~\eqref{eq:Nunb} and \eqref{eq:sigmath}, we obtain
\begin{align}
	\theta &= \int_{\Omega}\hat{\theta} dP^\mathsf{N}_{\frac{1}{2}\mathbb{I} + \frac{e^{i\theta}}{2} \ketbra{1}{0}+ \frac{e^{-i\theta}}{2} \ketbra{0}{1}} \notag\\
	&= A + Be^{i\theta} + Ce^{-i\theta}
	\label{eq:thcontra}
\end{align}
for all $\theta \in (\alpha ,\beta)$.
The RHS of Eq.~\eqref{eq:thcontra} is an entire function of $\theta$ and has the following Taylor expansion:
\begin{equation}
	A+B+C + i\theta (B-C) + \frac{(i\theta)^2}{2!} (B+C) + \frac{(i\theta)^3}{3!} (B-C) +\dotsb.
\end{equation}
By comparing the first- and third-order terms of $\theta$ we obtain
\begin{equation}
	1 = i(B-C) , \quad 0 = B-C ,
\end{equation}
which is a contradiction.
Therefore, $(P^\mathsf{N}_{\sigma_{\theta}})_{\theta \in \Theta}$ has no unbiased estimator for $\theta$.
\end{proof}
\bibliography{apssamp}

\providecommand{\noopsort}[1]{}\providecommand{\singleletter}[1]{#1}%
\begin{thebibliography}{20}%
\makeatletter
\providecommand \@ifxundefined [1]{%
 \@ifx{#1\undefined}
}%
\providecommand \@ifnum [1]{%
 \ifnum #1\expandafter \@firstoftwo
 \else \expandafter \@secondoftwo
 \fi
}%
\providecommand \@ifx [1]{%
 \ifx #1\expandafter \@firstoftwo
 \else \expandafter \@secondoftwo
 \fi
}%
\providecommand \natexlab [1]{#1}%
\providecommand \enquote  [1]{``#1''}%
\providecommand \bibnamefont  [1]{#1}%
\providecommand \bibfnamefont [1]{#1}%
\providecommand \citenamefont [1]{#1}%
\providecommand \href@noop [0]{\@secondoftwo}%
\providecommand \href [0]{\begingroup \@sanitize@url \@href}%
\providecommand \@href[1]{\@@startlink{#1}\@@href}%
\providecommand \@@href[1]{\endgroup#1\@@endlink}%
\providecommand \@sanitize@url [0]{\catcode `\\12\catcode `\$12\catcode
  `\&12\catcode `\#12\catcode `\^12\catcode `\_12\catcode `\%12\relax}%
\providecommand \@@startlink[1]{}%
\providecommand \@@endlink[0]{}%
\providecommand \url  [0]{\begingroup\@sanitize@url \@url }%
\providecommand \@url [1]{\endgroup\@href {#1}{\urlprefix }}%
\providecommand \urlprefix  [0]{URL }%
\providecommand \Eprint [0]{\href }%
\providecommand \doibase [0]{https://doi.org/}%
\providecommand \selectlanguage [0]{\@gobble}%
\providecommand \bibinfo  [0]{\@secondoftwo}%
\providecommand \bibfield  [0]{\@secondoftwo}%
\providecommand \translation [1]{[#1]}%
\providecommand \BibitemOpen [0]{}%
\providecommand \bibitemStop [0]{}%
\providecommand \bibitemNoStop [0]{.\EOS\space}%
\providecommand \EOS [0]{\spacefactor3000\relax}%
\providecommand \BibitemShut  [1]{\csname bibitem#1\endcsname}%
\let\auto@bib@innerbib\@empty
\bibitem [{\citenamefont {Takamoto}\ \emph {et~al.}(2020)\citenamefont
  {Takamoto}, \citenamefont {Ushijima}, \citenamefont {Ohmae}, \citenamefont
  {Yahagi}, \citenamefont {Kokado}, \citenamefont {Shinkai},\ and\
  \citenamefont {Katori}}]{Takamoto_2020}%
  \BibitemOpen
  \bibfield  {author} {\bibinfo {author} {\bibfnamefont {M.}~\bibnamefont
  {Takamoto}}, \bibinfo {author} {\bibfnamefont {I.}~\bibnamefont {Ushijima}},
  \bibinfo {author} {\bibfnamefont {N.}~\bibnamefont {Ohmae}}, \bibinfo
  {author} {\bibfnamefont {T.}~\bibnamefont {Yahagi}}, \bibinfo {author}
  {\bibfnamefont {K.}~\bibnamefont {Kokado}}, \bibinfo {author} {\bibfnamefont
  {H.}~\bibnamefont {Shinkai}},\ and\ \bibinfo {author} {\bibfnamefont
  {H.}~\bibnamefont {Katori}},\ }\bibfield  {title} {\bibinfo {title} {Test of
  general relativity by a pair of transportable optical lattice clocks},\
  }\href {https://doi.org/10.1038/s41566-020-0619-8} {\bibfield  {journal}
  {\bibinfo  {journal} {Nat. Photonics}\ }\textbf {\bibinfo {volume} {14}},\
  \bibinfo {pages} {411} (\bibinfo {year} {2020})}\BibitemShut {NoStop}%
\bibitem [{\citenamefont {Tanaka}\ and\ \citenamefont
  {Katori}(2021)}]{Tanaka_2021}%
  \BibitemOpen
  \bibfield  {author} {\bibinfo {author} {\bibfnamefont {Y.}~\bibnamefont
  {Tanaka}}\ and\ \bibinfo {author} {\bibfnamefont {H.}~\bibnamefont
  {Katori}},\ }\bibfield  {title} {\bibinfo {title} {Exploring potential
  applications of optical lattice clocks in a plate subduction zone},\ }\href
  {https://doi.org/10.1007/s00190-021-01548-y} {\bibfield  {journal} {\bibinfo
  {journal} {J. Geod.}\ }\textbf {\bibinfo {volume} {95}},\ \bibinfo {pages}
  {93} (\bibinfo {year} {2021})}\BibitemShut {NoStop}%
\bibitem [{\citenamefont {Zheng}\ \emph {et~al.}(2022)\citenamefont {Zheng},
  \citenamefont {Dolde}, \citenamefont {Lochab}, \citenamefont {Merriman},
  \citenamefont {Li},\ and\ \citenamefont {Kolkowitz}}]{Zheng_2022}%
  \BibitemOpen
  \bibfield  {author} {\bibinfo {author} {\bibfnamefont {X.}~\bibnamefont
  {Zheng}}, \bibinfo {author} {\bibfnamefont {J.}~\bibnamefont {Dolde}},
  \bibinfo {author} {\bibfnamefont {V.}~\bibnamefont {Lochab}}, \bibinfo
  {author} {\bibfnamefont {B.~N.}\ \bibnamefont {Merriman}}, \bibinfo {author}
  {\bibfnamefont {H.}~\bibnamefont {Li}},\ and\ \bibinfo {author}
  {\bibfnamefont {S.}~\bibnamefont {Kolkowitz}},\ }\bibfield  {title} {\bibinfo
  {title} {Differential clock comparisons with a multiplexed optical lattice
  clock},\ }\href {https://doi.org/10.1038/s41586-021-04344-y} {\bibfield
  {journal} {\bibinfo  {journal} {Nature}\ }\textbf {\bibinfo {volume} {602}},\
  \bibinfo {pages} {425} (\bibinfo {year} {2022})}\BibitemShut {NoStop}%
\bibitem [{\citenamefont {Bothwell}\ \emph {et~al.}(2022)\citenamefont
  {Bothwell}, \citenamefont {Kennedy}, \citenamefont {Aeppli}, \citenamefont
  {Kedar}, \citenamefont {Robinson}, \citenamefont {Oelker}, \citenamefont
  {Staron},\ and\ \citenamefont {Ye}}]{Bothwell_2022}%
  \BibitemOpen
  \bibfield  {author} {\bibinfo {author} {\bibfnamefont {T.}~\bibnamefont
  {Bothwell}}, \bibinfo {author} {\bibfnamefont {C.~J.}\ \bibnamefont
  {Kennedy}}, \bibinfo {author} {\bibfnamefont {A.}~\bibnamefont {Aeppli}},
  \bibinfo {author} {\bibfnamefont {D.}~\bibnamefont {Kedar}}, \bibinfo
  {author} {\bibfnamefont {J.~M.}\ \bibnamefont {Robinson}}, \bibinfo {author}
  {\bibfnamefont {E.}~\bibnamefont {Oelker}}, \bibinfo {author} {\bibfnamefont
  {A.}~\bibnamefont {Staron}},\ and\ \bibinfo {author} {\bibfnamefont
  {J.}~\bibnamefont {Ye}},\ }\bibfield  {title} {\bibinfo {title} {Resolving
  the gravitational redshift across a millimetre-scale atomic sample},\ }\href
  {https://doi.org/10.1038/s41586-021-04349-7} {\bibfield  {journal} {\bibinfo
  {journal} {Nature}\ }\textbf {\bibinfo {volume} {602}},\ \bibinfo {pages}
  {420} (\bibinfo {year} {2022})}\BibitemShut {NoStop}%
\bibitem [{\citenamefont {Pedrozo-Pe{\~{n}}afiel}\ \emph
  {et~al.}(2020)\citenamefont {Pedrozo-Pe{\~{n}}afiel}, \citenamefont
  {Colombo}, \citenamefont {Shu}, \citenamefont {Adiyatullin}, \citenamefont
  {Li}, \citenamefont {Mendez}, \citenamefont {Braverman}, \citenamefont
  {Kawasaki}, \citenamefont {Akamatsu}, \citenamefont {Xiao},\ and\
  \citenamefont {Vuleti{\'{c}}}}]{Pedrozo_2020}%
  \BibitemOpen
  \bibfield  {author} {\bibinfo {author} {\bibfnamefont {E.}~\bibnamefont
  {Pedrozo-Pe{\~{n}}afiel}}, \bibinfo {author} {\bibfnamefont {S.}~\bibnamefont
  {Colombo}}, \bibinfo {author} {\bibfnamefont {C.}~\bibnamefont {Shu}},
  \bibinfo {author} {\bibfnamefont {A.~F.}\ \bibnamefont {Adiyatullin}},
  \bibinfo {author} {\bibfnamefont {Z.}~\bibnamefont {Li}}, \bibinfo {author}
  {\bibfnamefont {E.}~\bibnamefont {Mendez}}, \bibinfo {author} {\bibfnamefont
  {B.}~\bibnamefont {Braverman}}, \bibinfo {author} {\bibfnamefont
  {A.}~\bibnamefont {Kawasaki}}, \bibinfo {author} {\bibfnamefont
  {D.}~\bibnamefont {Akamatsu}}, \bibinfo {author} {\bibfnamefont
  {Y.}~\bibnamefont {Xiao}},\ and\ \bibinfo {author} {\bibfnamefont
  {V.}~\bibnamefont {Vuleti{\'{c}}}},\ }\bibfield  {title} {\bibinfo {title}
  {Entanglement on an optical atomic-clock transition},\ }\href
  {https://doi.org/10.1038/s41586-020-3006-1} {\bibfield  {journal} {\bibinfo
  {journal} {Nature}\ }\textbf {\bibinfo {volume} {588}},\ \bibinfo {pages}
  {414} (\bibinfo {year} {2020})}\BibitemShut {NoStop}%
\bibitem [{\citenamefont {Kawasaki}(2022)}]{Kawasaki_2022}%
  \BibitemOpen
  \bibfield  {author} {\bibinfo {author} {\bibfnamefont {A.}~\bibnamefont
  {Kawasaki}},\ }\bibfield  {title} {\bibinfo {title} {Decoherence of atomic
  ensembles in optical lattice clocks by gravity},\ }\href
  {https://doi.org/10.7566/JPSJ.91.034301} {\bibfield  {journal} {\bibinfo
  {journal} {J. Phys. Soc. Jpn.}\ }\textbf {\bibinfo {volume} {91}},\ \bibinfo
  {pages} {034301} (\bibinfo {year} {2022})}\BibitemShut {NoStop}%
\bibitem [{\citenamefont {Pikovski}\ \emph {et~al.}(2015)\citenamefont
  {Pikovski}, \citenamefont {Zych}, \citenamefont {Costa},\ and\ \citenamefont
  {Brukner}}]{Pikovski_2015}%
  \BibitemOpen
  \bibfield  {author} {\bibinfo {author} {\bibfnamefont {I.}~\bibnamefont
  {Pikovski}}, \bibinfo {author} {\bibfnamefont {M.}~\bibnamefont {Zych}},
  \bibinfo {author} {\bibfnamefont {F.}~\bibnamefont {Costa}},\ and\ \bibinfo
  {author} {\bibfnamefont {{\v{C}}.}~\bibnamefont {Brukner}},\ }\bibfield
  {title} {\bibinfo {title} {Universal decoherence due to gravitational time
  dilation},\ }\href {https://doi.org/10.1038/nphys3366} {\bibfield  {journal}
  {\bibinfo  {journal} {Nat. Phys.}\ }\textbf {\bibinfo {volume} {11}},\
  \bibinfo {pages} {668} (\bibinfo {year} {2015})}\BibitemShut {NoStop}%
\bibitem [{\citenamefont {Ruiz}\ \emph {et~al.}(2017)\citenamefont {Ruiz},
  \citenamefont {Giacomini},\ and\ \citenamefont {Brukner}}]{Esteban_2017}%
  \BibitemOpen
  \bibfield  {author} {\bibinfo {author} {\bibfnamefont {E.~C.}\ \bibnamefont
  {Ruiz}}, \bibinfo {author} {\bibfnamefont {F.}~\bibnamefont {Giacomini}},\
  and\ \bibinfo {author} {\bibfnamefont {{\v{C}}.}~\bibnamefont {Brukner}},\
  }\bibfield  {title} {\bibinfo {title} {Entanglement of quantum clocks through
  gravity},\ }\href {https://doi.org/10.1073/pnas.1616427114} {\bibfield
  {journal} {\bibinfo  {journal} {Proc. Natl. Acad. Sci. U.S.A.}\ }\textbf
  {\bibinfo {volume} {114}},\ \bibinfo {pages} {E2303} (\bibinfo {year}
  {2017})}\BibitemShut {NoStop}%
\bibitem [{\citenamefont {Cepollaro}\ \emph {et~al.}(2023)\citenamefont
  {Cepollaro}, \citenamefont {Giacomini},\ and\ \citenamefont
  {Paris}}]{Cepollaro_2023}%
  \BibitemOpen
  \bibfield  {author} {\bibinfo {author} {\bibfnamefont {C.}~\bibnamefont
  {Cepollaro}}, \bibinfo {author} {\bibfnamefont {F.}~\bibnamefont
  {Giacomini}},\ and\ \bibinfo {author} {\bibfnamefont {M.~G.}\ \bibnamefont
  {Paris}},\ }\bibfield  {title} {\bibinfo {title} {Gravitational time dilation
  as a resource in quantum sensing},\ }\href
  {https://doi.org/10.22331/q-2023-03-13-946} {\bibfield  {journal} {\bibinfo
  {journal} {Quantum}\ }\textbf {\bibinfo {volume} {7}},\ \bibinfo {pages}
  {946} (\bibinfo {year} {2023})}\BibitemShut {NoStop}%
\bibitem [{\citenamefont {Yamaguchi}\ \emph {et~al.}(2019)\citenamefont
  {Yamaguchi}, \citenamefont {Safronova}, \citenamefont {Gibble},\ and\
  \citenamefont {Katori}}]{Yamaguchi_2019}%
  \BibitemOpen
  \bibfield  {author} {\bibinfo {author} {\bibfnamefont {A.}~\bibnamefont
  {Yamaguchi}}, \bibinfo {author} {\bibfnamefont {M.~S.}\ \bibnamefont
  {Safronova}}, \bibinfo {author} {\bibfnamefont {K.}~\bibnamefont {Gibble}},\
  and\ \bibinfo {author} {\bibfnamefont {H.}~\bibnamefont {Katori}},\
  }\bibfield  {title} {\bibinfo {title} {Narrow-line cooling and determination
  of the magic wavelength of cd},\ }\href
  {https://doi.org/10.1103/PhysRevLett.123.113201} {\bibfield  {journal}
  {\bibinfo  {journal} {Phys. Rev. Lett.}\ }\textbf {\bibinfo {volume} {123}},\
  \bibinfo {pages} {113201} (\bibinfo {year} {2019})}\BibitemShut {NoStop}%
\bibitem [{\citenamefont {Ludlow}\ \emph {et~al.}(2008)\citenamefont {Ludlow},
  \citenamefont {Zelevinsky}, \citenamefont {Campbell}, \citenamefont {Blatt},
  \citenamefont {Boyd}, \citenamefont {de~Miranda}, \citenamefont {Martin},
  \citenamefont {Thomsen}, \citenamefont {Foreman}, \citenamefont {Ye},
  \citenamefont {Fortier}, \citenamefont {Stalnaker}, \citenamefont {Diddams},
  \citenamefont {Coq}, \citenamefont {Barber}, \citenamefont {Poli},
  \citenamefont {Lemke}, \citenamefont {Beck},\ and\ \citenamefont
  {Oates}}]{Ludlow_2008}%
  \BibitemOpen
  \bibfield  {author} {\bibinfo {author} {\bibfnamefont {A.~D.}\ \bibnamefont
  {Ludlow}}, \bibinfo {author} {\bibfnamefont {T.}~\bibnamefont {Zelevinsky}},
  \bibinfo {author} {\bibfnamefont {G.~K.}\ \bibnamefont {Campbell}}, \bibinfo
  {author} {\bibfnamefont {S.}~\bibnamefont {Blatt}}, \bibinfo {author}
  {\bibfnamefont {M.~M.}\ \bibnamefont {Boyd}}, \bibinfo {author}
  {\bibfnamefont {M.~H.~G.}\ \bibnamefont {de~Miranda}}, \bibinfo {author}
  {\bibfnamefont {M.~J.}\ \bibnamefont {Martin}}, \bibinfo {author}
  {\bibfnamefont {J.~W.}\ \bibnamefont {Thomsen}}, \bibinfo {author}
  {\bibfnamefont {S.~M.}\ \bibnamefont {Foreman}}, \bibinfo {author}
  {\bibfnamefont {J.}~\bibnamefont {Ye}}, \bibinfo {author} {\bibfnamefont
  {T.~M.}\ \bibnamefont {Fortier}}, \bibinfo {author} {\bibfnamefont {J.~E.}\
  \bibnamefont {Stalnaker}}, \bibinfo {author} {\bibfnamefont {S.~A.}\
  \bibnamefont {Diddams}}, \bibinfo {author} {\bibfnamefont {Y.~L.}\
  \bibnamefont {Coq}}, \bibinfo {author} {\bibfnamefont {Z.~W.}\ \bibnamefont
  {Barber}}, \bibinfo {author} {\bibfnamefont {N.}~\bibnamefont {Poli}},
  \bibinfo {author} {\bibfnamefont {N.~D.}\ \bibnamefont {Lemke}}, \bibinfo
  {author} {\bibfnamefont {K.~M.}\ \bibnamefont {Beck}},\ and\ \bibinfo
  {author} {\bibfnamefont {C.~W.}\ \bibnamefont {Oates}},\ }\bibfield  {title}
  {\bibinfo {title} {Sr lattice clock at $1\times10^{-16}$ fractional
  uncertainty by remote optical evaluation with a ca clock},\ }\href
  {https://doi.org/10.1126/science.1153341} {\bibfield  {journal} {\bibinfo
  {journal} {Science}\ }\textbf {\bibinfo {volume} {319}},\ \bibinfo {pages}
  {1805} (\bibinfo {year} {2008})}\BibitemShut {NoStop}%
\bibitem [{\citenamefont {Barber}\ \emph {et~al.}(2008)\citenamefont {Barber},
  \citenamefont {Stalnaker}, \citenamefont {Lemke}, \citenamefont {Poli},
  \citenamefont {Oates}, \citenamefont {Fortier}, \citenamefont {Diddams},
  \citenamefont {Hollberg}, \citenamefont {Hoyt}, \citenamefont
  {Taichenachev},\ and\ \citenamefont {Yudin}}]{Barber_2008}%
  \BibitemOpen
  \bibfield  {author} {\bibinfo {author} {\bibfnamefont {Z.~W.}\ \bibnamefont
  {Barber}}, \bibinfo {author} {\bibfnamefont {J.~E.}\ \bibnamefont
  {Stalnaker}}, \bibinfo {author} {\bibfnamefont {N.~D.}\ \bibnamefont
  {Lemke}}, \bibinfo {author} {\bibfnamefont {N.}~\bibnamefont {Poli}},
  \bibinfo {author} {\bibfnamefont {C.~W.}\ \bibnamefont {Oates}}, \bibinfo
  {author} {\bibfnamefont {T.~M.}\ \bibnamefont {Fortier}}, \bibinfo {author}
  {\bibfnamefont {S.~A.}\ \bibnamefont {Diddams}}, \bibinfo {author}
  {\bibfnamefont {L.}~\bibnamefont {Hollberg}}, \bibinfo {author}
  {\bibfnamefont {C.~W.}\ \bibnamefont {Hoyt}}, \bibinfo {author}
  {\bibfnamefont {A.~V.}\ \bibnamefont {Taichenachev}},\ and\ \bibinfo {author}
  {\bibfnamefont {V.~I.}\ \bibnamefont {Yudin}},\ }\bibfield  {title} {\bibinfo
  {title} {Optical lattice induced light shifts in an yb atomic clock},\ }\href
  {https://doi.org/10.1103/PhysRevLett.100.103002} {\bibfield  {journal}
  {\bibinfo  {journal} {Phys. Rev. Lett.}\ }\textbf {\bibinfo {volume} {100}},\
  \bibinfo {pages} {103002} (\bibinfo {year} {2008})}\BibitemShut {NoStop}%
\bibitem [{\citenamefont {McFerran}\ \emph {et~al.}(2012)\citenamefont
  {McFerran}, \citenamefont {Yi}, \citenamefont {Mejri}, \citenamefont
  {Di~Manno}, \citenamefont {Zhang}, \citenamefont {Gu\'ena}, \citenamefont
  {Le~Coq},\ and\ \citenamefont {Bize}}]{McFerran_2012}%
  \BibitemOpen
  \bibfield  {author} {\bibinfo {author} {\bibfnamefont {J.~J.}\ \bibnamefont
  {McFerran}}, \bibinfo {author} {\bibfnamefont {L.}~\bibnamefont {Yi}},
  \bibinfo {author} {\bibfnamefont {S.}~\bibnamefont {Mejri}}, \bibinfo
  {author} {\bibfnamefont {S.}~\bibnamefont {Di~Manno}}, \bibinfo {author}
  {\bibfnamefont {W.}~\bibnamefont {Zhang}}, \bibinfo {author} {\bibfnamefont
  {J.}~\bibnamefont {Gu\'ena}}, \bibinfo {author} {\bibfnamefont
  {Y.}~\bibnamefont {Le~Coq}},\ and\ \bibinfo {author} {\bibfnamefont
  {S.}~\bibnamefont {Bize}},\ }\bibfield  {title} {\bibinfo {title} {Neutral
  atom frequency reference in the deep ultraviolet with
  $\mathrm{\text{fractional
  uncertainty}}=5.7\ifmmode\times\else\texttimes\fi{}{10}^{\ensuremath{-}15}$},\
  }\href {https://doi.org/10.1103/PhysRevLett.108.183004} {\bibfield  {journal}
  {\bibinfo  {journal} {Phys. Rev. Lett.}\ }\textbf {\bibinfo {volume} {108}},\
  \bibinfo {pages} {183004} (\bibinfo {year} {2012})}\BibitemShut {NoStop}%
\bibitem [{\citenamefont {Kulosa}\ \emph {et~al.}(2015)\citenamefont {Kulosa},
  \citenamefont {Fim}, \citenamefont {Zipfel}, \citenamefont {R\"uhmann},
  \citenamefont {Sauer}, \citenamefont {Jha}, \citenamefont {Gibble},
  \citenamefont {Ertmer}, \citenamefont {Rasel}, \citenamefont {Safronova},
  \citenamefont {Safronova},\ and\ \citenamefont {Porsev}}]{Kulosa_2015}%
  \BibitemOpen
  \bibfield  {author} {\bibinfo {author} {\bibfnamefont {A.~P.}\ \bibnamefont
  {Kulosa}}, \bibinfo {author} {\bibfnamefont {D.}~\bibnamefont {Fim}},
  \bibinfo {author} {\bibfnamefont {K.~H.}\ \bibnamefont {Zipfel}}, \bibinfo
  {author} {\bibfnamefont {S.}~\bibnamefont {R\"uhmann}}, \bibinfo {author}
  {\bibfnamefont {S.}~\bibnamefont {Sauer}}, \bibinfo {author} {\bibfnamefont
  {N.}~\bibnamefont {Jha}}, \bibinfo {author} {\bibfnamefont {K.}~\bibnamefont
  {Gibble}}, \bibinfo {author} {\bibfnamefont {W.}~\bibnamefont {Ertmer}},
  \bibinfo {author} {\bibfnamefont {E.~M.}\ \bibnamefont {Rasel}}, \bibinfo
  {author} {\bibfnamefont {M.~S.}\ \bibnamefont {Safronova}}, \bibinfo {author}
  {\bibfnamefont {U.~I.}\ \bibnamefont {Safronova}},\ and\ \bibinfo {author}
  {\bibfnamefont {S.~G.}\ \bibnamefont {Porsev}},\ }\bibfield  {title}
  {\bibinfo {title} {Towards a mg lattice clock: Observation of the
  $^{1}{S}_{0}\text{\ensuremath{-}}^{3}{P}_{0}$ transition and determination of
  the magic wavelength},\ }\href
  {https://doi.org/10.1103/PhysRevLett.115.240801} {\bibfield  {journal}
  {\bibinfo  {journal} {Phys. Rev. Lett.}\ }\textbf {\bibinfo {volume} {115}},\
  \bibinfo {pages} {240801} (\bibinfo {year} {2015})}\BibitemShut {NoStop}%
\bibitem [{\citenamefont {Shao}(2003)}]{shao2003}%
  \BibitemOpen
  \bibfield  {author} {\bibinfo {author} {\bibfnamefont {J.}~\bibnamefont
  {Shao}},\ }\href@noop {} {\emph {\bibinfo {title} {Mathematical
  statistics}}},\ \bibinfo {edition} {2nd}\ ed.,\ Springer texts in statistics\
  (\bibinfo  {publisher} {Springer},\ \bibinfo {year} {2003})\BibitemShut
  {NoStop}%
\bibitem [{\citenamefont {Petz}\ and\ \citenamefont
  {Ghinea}(2011)}]{PETZ_2011}%
  \BibitemOpen
  \bibfield  {author} {\bibinfo {author} {\bibfnamefont {D.}~\bibnamefont
  {Petz}}\ and\ \bibinfo {author} {\bibfnamefont {C.}~\bibnamefont {Ghinea}},\
  }\bibfield  {title} {\bibinfo {title} {Introduction to quantum {Fisher}
  information},\ }in\ \href {https://doi.org/10.1142/9789814338745_0015} {\emph
  {\bibinfo {booktitle} {Quantum Probability and Related Topics}}}\ (\bibinfo
  {publisher} {World Scientific},\ \bibinfo {year} {2011})\BibitemShut
  {NoStop}%
\bibitem [{\citenamefont {Helstrom}(1976)}]{helstrom1976}%
  \BibitemOpen
  \bibfield  {author} {\bibinfo {author} {\bibfnamefont {C.~W.}\ \bibnamefont
  {Helstrom}},\ }\href@noop {} {\emph {\bibinfo {title} {Quantum detection and
  estimation theory}}},\ \bibinfo {series} {Mathematics in science and
  engineering : a series of monographs and textbooks}\ No.\ \bibinfo {number}
  {123}\ (\bibinfo  {publisher} {Academic Press},\ \bibinfo {year}
  {1976})\BibitemShut {NoStop}%
\bibitem [{\citenamefont {Holevo}(2011)}]{holevo2011}%
  \BibitemOpen
  \bibfield  {author} {\bibinfo {author} {\bibfnamefont {A.~S.}\ \bibnamefont
  {Holevo}},\ }\href@noop {} {\emph {\bibinfo {title} {Probabilistic and
  statistical aspects of quantum theory}}},\ \bibinfo {series} {Quaderni
  monographs}\ No.~\bibinfo {number} {1}\ (\bibinfo  {publisher} {Edizioni
  della Normale},\ \bibinfo {year} {2011})\BibitemShut {NoStop}%
\bibitem [{\citenamefont {Watanabe}(2014)}]{Watanabe_2014}%
  \BibitemOpen
  \bibfield  {author} {\bibinfo {author} {\bibfnamefont {Y.}~\bibnamefont
  {Watanabe}},\ }in\ \href@noop {} {\emph {\bibinfo {booktitle} {Formulation of
  Uncertainty Relation Between Error and Disturbance in Quantum Measurement by
  Using Quantum Estimation Theory}}}\ (\bibinfo  {publisher} {Springer Tokyo},\
  \bibinfo {year} {2014})\ Chap.~\bibinfo {chapter} {4}, pp.\ \bibinfo {pages}
  {37--44},\ \bibinfo {edition} {1st}\ ed.\BibitemShut {Stop}%
\bibitem [{\citenamefont {Petz}(1996)}]{petz1996}%
  \BibitemOpen
  \bibfield  {author} {\bibinfo {author} {\bibfnamefont {D.}~\bibnamefont
  {Petz}},\ }\bibfield  {title} {\bibinfo {title} {Monotone metrics on matrix
  spaces},\ }\href
  {https://doi.org/https://doi.org/10.1016/0024-3795(94)00211-8} {\bibfield
  {journal} {\bibinfo  {journal} {Linear Algebra and its Applications}\
  }\textbf {\bibinfo {volume} {244}},\ \bibinfo {pages} {81} (\bibinfo {year}
  {1996})}\BibitemShut {NoStop}%
\end{thebibliography}%
\end{document}